\let\oldvec\vec

\documentclass[runningheads, envcountsame, a4paper]{llncs}

\usepackage[utf8]{inputenc}
\usepackage{bbold}
\usepackage{cite}
\usepackage{amsmath}
\usepackage{amssymb}
\usepackage[x11names]{xcolor}
\usepackage{microtype}
\usepackage{mathpartir}
\usepackage{hyperref}

\let\vec\oldvec

\newcommand\LQ[1]{#1\backslash}

\newcommand{\True}{\ensuremath{\mathbf{tt}}}
\newcommand{\False}{\ensuremath{\mathbf{ff}}}

\newcommand{\DOM}{{\textit{dom}}}

\newcommand\IPD{\ensuremath\Delta}
\newcommand\Power{\wp}
\newcommand\nat{\mathbb{N}}
\newcommand\bool{\mathbb{B}}

\newcommand\pderiv[3][{}]{\partial^{#1}_{#3}(#2)}
\newcommand\cderiv[3][{}]{\tilde\partial^{#1}_{#3}(#2)}

\newcommand\Angle[1]{\langle#1\rangle}
\newcommand\LFP{\textsf{lfp}}
\newcommand\FV{{\textit{fv}}}

\newcommand\Rnull{\mathbf0}
\newcommand\Rempty{\mathbf1}

\newcommand\Lang[1][{}]{\mathcal{L}^{#1}}
\renewcommand\L\Lang

\newcommand\X{\mathcal{X}}

\newcommand\calA{\mathcal{A}}

\newcommand\Reg{\mathbf{R}}

\newcommand\PUSH{:}
\newcommand\EMPTY{[\,]}
\newcommand\SINGLETON[1]{{[#1]}}
\newcommand\Null{\mathcal{N}}

\newcommand\UA{\ensuremath{\mathcal{U}\mathcal{A}}}

\newcommand\ADDRESS{A}

\newcommand\RS{\ensuremath{\overline{\mathbf{r}}}}
\renewcommand\SS{\ensuremath{\overline{\mathbf{s}}}}

\newcommand\ApplySubst[2]{#1 \bullet #2} 

\newcommand\Ruleform[1]{{\setlength{\fboxrule}{1pt}\fbox{\normalsize $#1$}}}

\newcommand\partialto\hookrightarrow

\newcommand\OccursBefore\preceq
\newcommand\OccursStrictlyBefore\prec

\newcommand\StepsTo\vdash

\title{Partial Derivatives for Context-Free Languages}
\subtitle{From $\mu$-Regular Expressions to Pushdown
  Automata\footnote{Full version with proofs is available at \url{https://arxiv.org/abs/1610.06832}.}}
\author{Peter Thiemann}
\institute{University of Freiburg}

\begin{document}

\maketitle
\setcounter{footnote}{0}
\begin{abstract}
  We extend Antimirov's partial derivatives from regular expressions
  to $\mu$-regular expressions that describe context-free languages.
  We prove the correctness of partial derivatives as well as the
  finiteness of the set of iterated partial derivatives. The latter
  are used as pushdown symbols in our construction of a
  nondeterministic pushdown automaton, which generalizes Antimirov's NFA
  construction.
\end{abstract}
\section{Introduction}
\label{sec:introduction}

Brzozowski derivatives \cite{321249} and Antimirov's partial derivatives \cite{Antimirov96Partial} are
well-known tools to transform  regular expressions to finite automata and to define algorithms for
equivalence and containment of regular languages \cite{AntimirovIneq,Grabmayer:2005:UPC:2156157.2156171}.
Both automata constructions rely on the finiteness of the
set of iterated derivatives. Brzozowski derivatives need to be considered up to similarity
(commutativity, associativity, and idempotence for union) to obtain finiteness.  
Derivatives had quite some impact on the study of
algorithms for regular languages on finite words and trees \cite{DBLP:conf/rta/RouV03,CaronChamparnaudMignot2011}.

There are many studies of derivative structures for enhancements of regular expressions.
While Brzozowski's original work covered extended regular
expressions, partial derivatives were originally limited to simple expressions without intersection and
complement. It is a significant effort to define partial derivatives for extended regular expressions
\cite{CaronChamparnaudMignot2011}. Many further operators have been
considered, among them shuffle operators
\cite{DBLP:conf/lata/SulzmannT15},  multi-tilde-bar
expressions  \cite{DBLP:conf/wia/CaronCM12}, expressions with
multiplicities \cite{DBLP:journals/tcs/LombardyS05}, approximate
regular expressions  \cite{DBLP:conf/lata/ChamparnaudJM12}, and many
more. There have been a number of approaches to develop general
frameworks for derivation: Caron and coworkers
\cite{DBLP:journals/ita/CaronCM14} abstract over the support for
creating derivations, Thiemann \cite{DBLP:conf/wia/Thiemann16}
develops criteria for derivable language operators.

Recently, there has been practical interest in the study of
derivatives and partial derivatives. Owens and coworkers \cite{re-derivs}
report a functional implementation with some extensions (e.g., character classes) to handle large
character sets, which is partially rediscovering work on the FIRE
library \cite{DBLP:conf/wia/Watson96}. Might and coworkers
\cite{DBLP:conf/icfp/MightDS11,DBLP:conf/pldi/0001HM16} push beyond
regular languages by implementing parsing for
context-free languages using derivatives and demonstrate its
feasibility in practice. 

Winter and coworkers \cite{DBLP:conf/calco/WinterBR11} study
context-free languages in a coalgebraic setting. 
They use a notion of derivative to give
three equivalent characterizations of context-free languages by
grammars in weak Greibach normal form, behavioral differential
equations, and guarded $\mu$-regular expressions. 

In this work, we focus on using derivatives for parsing of
context-free languages. While Might and coworkers explore algorithmic
issues, we investigate the correctness of context-free parsing
with derivatives. To this end, we develop the theory of derivatives for $\mu$-regular expressions,
which extend regular expressions with a least fixed point operator. Our results are relevant for
context-free parsing because $\mu$-regular expressions are equivalent
to context-free grammars in generating power. Compared to the
work of Winter and coworkers \cite{DBLP:conf/calco/WinterBR11}, we do not require recursion to
be guarded (i.e., we admit left recursion) and we focus on establishing the connection to pushdown
automata. Unguarded recursion forces us to consider derivation by 
$\varepsilon$, which corresponds to an unfolding of a left-recursive
$\mu$-expression. Guarded expressions always admit a proper derivation
by a symbol.

Our theory is the proper generalization of Antimirov's theory of
partial derivatives to $\mu$-regular expressions: our derivative function 
corresponds \emph{exactly} to the transition function of the nondeterministic
pushdown automaton that recognizes the same language.
The pendant of Antimirov's finiteness result yields the finiteness of
the set of pushdown symbols of this automaton.

\section{Preliminaries}
\label{sec:preliminaries}

We write $\nat$ for the set of natural numbers, $\bool = \{ \False, \True \}$ for the set of
booleans, and $X \uplus Y$ for the disjoint union of sets $X$ and
$Y$. We consider total maps $m:X \to Y$ as sets of pairs in the usual
way, so that $m \subseteq X \times Y$ and $\emptyset$ 
denotes the empty mapping. For $x_0\in X$ and $y_0\in Y$, the \emph{map update} of $m$ is defined as
$m[y_0/x_0] (x) = y_0$ if $x=x_0$ and $m[y_0/x_0] (x) = m (x)$ if $x \ne x_0$.

For conciseness, we fix a finite set of symbols, $\Sigma$,  as the
underlying \emph{alphabet}. We write $\Sigma^*$ for the set of finite words
over $\Sigma$, $\varepsilon\in\Sigma^*$ stands for the empty word, and $\Sigma^+ = \Sigma^*
\setminus \{\varepsilon\}$. For 
$u,v\in\Sigma^*$ , we write $|u|\in\nat$ for the length of $u$ and $u\cdot v $ (or just $uv$) for the
concatenation of words.

Given languages $U,V,W \subseteq \Sigma^*$, 
concatenation extends to languages as usual: $U
\cdot V = \{ u \cdot v \mid u\in U, v \in V\}$.
The Kleene closure is defined as the smallest set $U^* \subseteq \Sigma^*$ such that $U^* = \{\varepsilon\} \cup U
\cdot U^*$.
We write the \emph{left quotient} as $\LQ{U}W
= \{ v \mid v\in \Sigma^*, \exists u\in U: uv \in W \}$. For a singleton language $U = \{u\}$, we write $\LQ{u}W$
for the left quotient.

\begin{definition}
  A (nondeterministic) finite automaton (NFA) is a tuple $\mathcal{A} = (Q, \Sigma, \delta, q_0, F)$ where $Q$ is a finite
  set of states, $\Sigma$ an alphabet, $\delta \subseteq Q \times
  \Sigma \times Q$ the transition relation, $q_0 \in Q$ the initial state, and $F \subseteq Q$ the set of final states.

  Let $n\in\nat$.
  A \emph{run} of $\mathcal{A}$ on $w = a_0\dots a_{n-1}\in \Sigma^*$ is a sequence $q_0\dots q_n \in
  Q^*$ such that, for all $0\le i <n$,
  $(q_i, a_i, q_{i+1}) \in \delta$. The run is \emph{accepting} if $q_n \in F$.
  The language recognized by $\mathcal{A}$ is
  $\Lang (\mathcal{A}) = \{ w \in \Sigma^* \mid \exists
  \textrm{ accepting run of $\mathcal{A}$ on $w$} \}$.
  %
\end{definition}

\begin{definition}
  A (nondeterministic) pushdown automaton (PDA) is a tuple
  $\mathcal{P} = (Q, \Sigma, \Gamma, \delta, q_0, Z_0)$ where $Q$ is a
  finite set of states, $\Sigma$ the input alphabet, $\Gamma$ the
  pushdown alphabet (a finite set), $\delta \subseteq Q \times (\Sigma
  \cup \{\varepsilon\}) \times \Gamma \times Q \times \Gamma^* $ is
  the transition relation, $q_0
  \in Q$ is the initial state, $Z_0 \in \Gamma$ is the 
  bottom symbol.

  A \emph{configuration of $\mathcal{P} $} is a tuple $c \in Q \times
  \Sigma^* \times \Gamma^*$ of the current state, the rest of the
  input, and the current contents of the pushdown.

  The transition relation $\delta$ gives rise to a binary stepping relation
  $\vdash$ on configurations defined by (for all $q, q' \in Q$,
  $\alpha\in\Sigma\cup\{\varepsilon\}$, $Z\in\Gamma$, $\gamma,\gamma'
  \in \Gamma^*$, $v\in\Sigma^*$):
  \begin{mathpar}
    \inferrule{
      (q, \alpha, Z, q', \gamma') \in \delta
    }{
      (q, \alpha v, Z\gamma) \vdash (q', v, \gamma'\gamma)
    }
  \end{mathpar}
  The language of the PDA is $\Lang (\mathcal{P}) = \{ v \in \Sigma^* \mid
  \exists q \in Q: (q_0, v,  Z_0) \vdash^*  (q, \varepsilon,
  \varepsilon) \}$ where $\vdash^*$ is the reflexive transitive
  closure of $\vdash$.
\end{definition}

\section{$\mu$-Regular Expressions}
\label{sec:mu-regul-expr}

Regular expressions can be extended with a least fixed point operator
$\mu$ to extend their scope to context-free languages
\cite{DBLP:conf/csl/Leiss91}. 

\begin{definition}\label{def:regular-expression}
  The set $\Reg (\Sigma, X)$ of $\mu$-regular pre-expressions over alphabet
  $\Sigma$ and set of variables $X$ is defined as the smallest set such that
  \begin{itemize}
  \item $\Rnull \in \Reg (\Sigma, X)$,
  \item $\Rempty \in \Reg (\Sigma, X)$,
  \item $a\in \Sigma$ implies $a \in \Reg (\Sigma, X)$,
  \item $r, s \in \Reg (\Sigma, X)$ implies $r \cdot  s\in \Reg (\Sigma, X)$,
  \item $r,s \in \Reg (\Sigma, X)$ implies $r+s \in \Reg (\Sigma, X)$,
  \item $r \in \Reg (\Sigma, X)$ implies $r^*\in \Reg (\Sigma, X)$,
  \item $x \in X$ implies $x \in \Reg (\Sigma, X)$,
  \item $r \in \Reg (\Sigma, X \cup \{x\})$ implies $\mu x.r \in \Reg
    (\Sigma, X)$. 
  \end{itemize}

  The set $\Reg (\Sigma)$ of $\mu$-regular expressions over $\Sigma$ is
  defined as  $\Reg (\Sigma) := \Reg (\Sigma, \emptyset)$.
\end{definition}
As customary, we consider the elements of $\Reg (\Sigma,X)$ as abstract syntax trees and freely use
parentheses to disambiguate. We further assume that $*$ 
has higher precedence than $\cdot$, which has higher precedence than
$+$. The $\mu x$-operator binds the recursion variable $x$ with lowest precedence: its scope
extends as far to the right as possible. A variable $x$ occurs free if
there is no enclosing $\mu x$-operator.  A
\emph{closed} expression has no free variables.
\begin{definition}
  The language denoted by a $\mu$-regular pre-expression is defined
  inductively by $\Lang: \Reg(\Sigma, X) \times (X \to \Power (\Sigma^*))
  \to \Power(\Sigma^*)$. Let $\eta \in X \to \Power (\Sigma^*)$ be a
  mapping from variables to languages.
  \begin{itemize}
  \item $\Lang(\Rnull,  \eta ) = \{\}$.
  \item $\Lang(\Rempty,  \eta ) = \{\varepsilon\}$.
  \item $\Lang(a,  \eta ) = \{a\}$ (singleton letter word) for each
    $a\in\Sigma$.
  \item $\Lang(r\cdot s,  \eta ) = \Lang(r, \eta) \cdot \Lang(s, \eta)$.
  \item $\Lang(r+s,  \eta )  = \Lang(r,  \eta ) \cup \Lang(s,  \eta )$.
  \item $\Lang(r^*,  \eta ) = (\Lang(r,  \eta ))^*$.
  \item $\Lang (x,  \eta ) = \eta (x)$.
  \item $\Lang (\mu x. r,  \eta ) = \LFP\ L. \Lang (r, {\eta[x \mapsto L]}) $.
  \end{itemize}
  For an expression  $r \in \Reg (\Sigma)$,
  we write $\Lang (r) := \Lang (r,  \emptyset)$.
\end{definition}
Here, $\LFP$ is the \emph{least fixed point operator} on the complete
lattice $\Power (\Sigma^*)$ (ordered by set inclusion). Its
application in the definition yields
the smallest set $L \subseteq \Sigma^*$ such that $L = \Lang (r, {\eta[x \mapsto L]})
$. This fixed point exists by Tarski's theorem because $\Lang$ is a monotone
function, which is captured precisely in the following lemma.

\begin{lemma}\label{lemma:lang-is-monotone}
  For each finite set $X$, $\eta \in X \to \Power (\Sigma)$, $r \in \Reg
  (\Sigma, X \cup \{x\})$, the function
  $L \mapsto \Lang (r, { \eta[x \mapsto L]})$ is monotone on
  $\Power (\Sigma^*)$. That is, if $L\subseteq L'$, then $\Lang (r, {\eta[x \mapsto L]})
  \subseteq \Lang (r, { \eta[x \mapsto L']})$.
\end{lemma}

According to Leiss \cite{DBLP:conf/csl/Leiss91}, it is a folkore
theorem that the languages generated by $\mu$-regular expressions are
exactly the context-free languages.

\begin{theorem}
  $L\subseteq \Sigma^*$ is context-free if and only if there exists a
  $\mu$-regular expression $r \in \Reg (\Sigma)$ such that $L = \Lang (r)$.
\end{theorem}

Subsequently we will deal syntactically with fixed points. To this end,
we define properties of expressions and substitutions to make
substitution application well-defined.

\begin{definition}
  Let $\X$ be the universe of variables occurring in expressions
  equipped with a strict partial order $\prec$.

  An expression is \emph{order-respecting} if each subexpression
  of the form $\mu x.r$ has only free variables which are strictly before
  $x$: $\forall y\in\FV (\mu x.r), y \prec x$.

  A mapping $\sigma: X \to \Reg (\Sigma, X)$ is \emph{order-closed}
  if $\forall x\in X$, $\sigma (x)$ is order-respecting and
  $\forall y \in \FV(\sigma (x))$, $y\prec x$ and $y\in\DOM
  (\sigma)$. 
\end{definition}
A variable ordering for an expression always exists: assume that all
binders bind different variables and take the topological
sort of the subexpression containment.

We define the application $\ApplySubst\sigma r$ of an order-closed
mapping $\sigma$ to an order-respecting expression $r$ by starting 
to substitute a maximal free variable by its image and repeat this
process until all variables are eliminated.
\begin{definition}
  Let $X \subseteq \X$ a finite set of variables, $r \in \Reg (\Sigma,
  X)$ order-respecting, and $\sigma : X \to
  \Reg (\Sigma, X)$ be order-closed.

  The application $\ApplySubst \sigma r \in 
  \Reg (\Sigma, X)$ yields an expression that is
  defined by substituting for the free variables in $r$ in descending order.
  \begin{align*}
    \ApplySubst \sigma r &=
                           \begin{cases}
                             r & \FV (r) = \emptyset \\
                             \ApplySubst\sigma{r[\sigma(x)/x]}& x \in \max (\FV (r)) \text{ is a
                               maximal element}
                           \end{cases}
  \end{align*}
\end{definition}
Application is well-defined because the variables $x$ are drawn from
the finite set $X$ and the substitution step for $x$ only introduces
new variables that are strictly smaller than $x$ due to
order-closedness.
The outcome does not depend on the choice of the maximal variable
because the unfolding of a maximal variable cannot contain one of the other
maximal variables. 
Furthermore, all intermediate expressions (and thus the result) are order-respecting.

\section{Partial Derivatives}
\label{sec:partial-derivatives}

\begin{figure}[t]
  \begin{minipage}[t]{0.6\linewidth}
  \begin{align*}
    \pderiv{\Rnull}{a} &= \{\} \\
    \pderiv{\Rempty}{a} &= \{\} \\
    \pderiv{b}{a} &= \{ \Rempty \mid a=b \}    \\
    \pderiv{r + s}{a} &= \pderiv{r}{a} \cup \pderiv{s}{a} \\
    \pderiv{r \cdot s}{a} &= \pderiv{r}{a} \cdot s \cup \{ s' \mid \Null
                            (r), s' \in \pderiv{s}{a} \} \\
    \pderiv{r^*}{a} &= \pderiv{r}{a} \cdot r^*
  \end{align*}
  \end{minipage}
  \begin{minipage}[t]{0.4\linewidth}
  \begin{align*}
    \Null (\Rnull) &= \False \\
    \Null (\Rempty) &= \True \\
    \Null (a) &= \False \\
    \Null (r + s) &= \Null (r) \vee \Null (s) \\
    \Null (r \cdot s) &= \Null (r) \wedge \Null (s) \\
    \Null (r^*) &= \True 
  \end{align*}
  \end{minipage}
    \color{black}
  \caption{Antimirov's definition of partial derivatives and nullability}
  \label{fig:partial-derivatives}
\end{figure}

Antimirov \cite{Antimirov96Partial} introduced partial derivatives to
study the syntactic transformation from regular expressions to
nondeterministic and deterministic finite automata. A partial
derivative $\pderiv{r}{a}$ with respect to an input symbol $a$ maps an
expression $r$ to a set of expressions such that their union denotes
the left quotient of $\Lang (r)$. Antimirov's definition corresponds
to the left part of Figure~\ref{fig:partial-derivatives}. We write $\Reg_o (\Sigma)$ for
the set of ordinary regular expressions that neither contain the
$\mu$-operator nor any variables.
We extend $\cdot$ to a function $(\cdot) : \Power (\Reg (\Sigma,X))
\times \Reg (\Sigma,X) \to \Power (\Reg (\Sigma,X))$ on sets of
expressions $R$ defined pointwise by
\begin{align*}
  R \cdot s &= \{ r \cdot s \mid r \in R \}
              \text.
\end{align*}

The definition of partial derivatives relies on nullability, which is tested by a function $\Null
: \Reg_o (\Sigma) \to \bool$. The right side of the figure corresponds to
Antimirov's definition.
\begin{lemma}
  For all $r\in\Reg_o (\Sigma)$, $\Null (r)$ iff $\varepsilon \in \Lang (r)$.
\end{lemma}
\begin{theorem}[Correctness \cite{Antimirov96Partial}]
  For all $r \in \Reg_o (\Sigma)$, $a\in\Sigma$,
  $\Lang ( \pderiv{r}{a}) = \LQ{a} \Lang (r)$.
\end{theorem}
Here we adopt the convention that if $R$ is a set of expressions, then
$\Lang ( R)$ denotes the union of the languages of all expressions:
$\Lang ( R) = \bigcup \{ \Lang (r) \mid r\in R\}$.
\begin{theorem}[Expansion]
  For $r \in \Reg_o (\Sigma)$, $\L (r) = \{ \varepsilon \mid \Null
  (r) \} \cup \bigcup_{a\in\Sigma} a\cdot \Lang (\pderiv{r}{a})$.
\end{theorem}

Partial derivatives give rise to a nondeterministic finite automaton.

\begin{theorem}[Finiteness \cite{Antimirov96Partial}]
  Let $r\in \Reg_o (\Sigma)$ be a regular expression.
  Define partial derivatives by words by $\pderiv{r}{\varepsilon} =
  \{r\}$ and $\pderiv{r}{aw} = \bigcup \{ \pderiv{s}{w} \mid s \in
  \pderiv{r}{a} \}$ and by a language $L$ by $\pderiv{r}{L} = \bigcup
  \{ \pderiv{r}{w} \mid w\in L\}$.

  The set $\pderiv{r}{\Sigma^*}$ is finite.
\end{theorem}

\begin{theorem}[Nondeterministic finite automaton construction \cite{Antimirov96Partial}]
  Let $r\in \Reg_o (\Sigma)$ be a regular expression and define
  $Q = \pderiv{r}{\Sigma^*}$,
  $\delta : Q \times \Sigma \to \Power (Q)$ by  $(q, a, q') \in
  \delta$ iff $q' \in\pderiv{q}{a}$.
  Let further $q_0 = r$ and $F = \{ q \in Q \mid \Null (q) \}$.

  Then  $\calA = (Q, \Sigma, \delta, q_0, F)$ is a NFA such that 
  $\L (r)= \L (\calA)$.
\end{theorem}

The plan is to extend these results to $\mu$-regular expressions. 
We start with the extension of the nullability function.

\section{Nullability}
\label{sec:nullability}

\begin{figure}[t]
  \begin{align*}
    \Null (\Rnull) \nu &= \False \\
    \Null (\Rempty) \nu &= \True \\
    \Null (a) \nu &= \False \\
    \Null (r + s) \nu &= \Null (r) \nu \vee \Null (s) \nu \\
    \Null (r \cdot s) \nu &= \Null (r)\nu \wedge \Null (s)\nu \\
    \Null (r^*) \nu &= \True \\
    \Null (\mu x. r) \nu & = \LFP\ b.\Null (r)\nu[x \mapsto b] \\
    \Null (x)\nu &= \nu (x)
  \end{align*}
  \caption{Nullability of $\mu$-regular expressions}
  \label{fig:nullability-mu-regular}
\end{figure}
Figure~\ref{fig:nullability-mu-regular} extends nullability to $\mu$-regular expressions. To cater
for recursion, the $\Null$ function obtains as a further argument a nullability environment $\nu$ of
type $X \to \bool$. With this extension, an  expression $\mu x.r$ is deemed nullable
if its body  $r$ is nullable. Furthermore, the least fixed point operator feeds back the nullability of
the body to the free occurrences of the recursion variables. This
fixed point is computed on the
two-element Boolean lattice $\bool$ ordered by $\False \sqsubseteq \True$ with
disjunction $(\vee) : \bool \times \bool \to \bool$ as the
least upper bound operation. Thus, the case for a free variable $x$
obtains its nullability information from the nullability environment. 
\begin{lemma}
  For each $r\in\Reg (\Sigma, X)$, $\Null (r)$ is a monotone function
  from $X \to \bool$ (ordered pointwise) to $\bool$.
\end{lemma}

To prepare for the correctness proof of $\Null$, we first simplify the
case for the fixed point. It turns out that one iteration is sufficient
to obtain the fixed point. This fact is also a consequence of a standard
result, namely that the number of iterations needed to compute the
fixed point of a monotone function on a lattice is bounded by the height
of the lattice. In this case, the Boolean lattice has height one.
\begin{lemma}\label{lemma:epsilon-in-empty-l}
  Let $X$ be a set of variables, $r \in \Reg (\Sigma, X \cup\{x\})$, 
  $\eta : X \to \Power (\Sigma^*)$, and $L\subseteq \Sigma^*$ such
  that $\varepsilon \notin L$.
  If $\varepsilon \notin\Lang (r, { \eta[x \mapsto
  \emptyset]}) $, then $\varepsilon \notin \Lang (r, { \eta[x \mapsto L]})$. 
\end{lemma}
\begin{lemma}\label{lemma:nullability-fixpoint}
  For all $r\in\Reg (\Sigma,X)$, for all $\nu : X\to \bool$,
  $$\LFP\ b.\Null (r)\nu[x \mapsto b] = \Null (r)\nu[x \mapsto \False].$$ 
\end{lemma}

For the statement of the correctness, we need to define what it means
for a nullability environment to agree with a language environment.
\begin{definition}
  Nullability environment $\nu : X \to \bool$ \emph{agrees} with
  language environment $\eta : X \to \Power (\Sigma^*)$, written
  $\eta \models \nu$, if for all $x\in X$, $\varepsilon \in \eta (x)$ iff $\nu (x)$.
\end{definition}

\begin{lemma}[Correctness of $\Null$]\label{lemma:correctness-of-null}
  For all $X$,  $r \in \Reg (\Sigma, X)$, $\eta\in X \to \Power (\Sigma^*)$, $\nu \in X \to \bool$,
  such that $\eta \models \nu$, it holds that 
  $\varepsilon \in \Lang (r, {\eta})$ iff $\Null (r)\nu$.
\end{lemma}


\section{Derivation}
\label{sec:derivation-v2}

\begin{figure}[t]
  \begin{align*}
    \pderiv[\sigma,\nu]{\Rnull}{\alpha} &= \{\} \\
    \pderiv[\sigma,\nu]{\Rempty}{\alpha} &= \{\} \\
    \pderiv[\sigma,\nu]{b}{\alpha} &= \{ [\Rempty] \mid \alpha=b \in\Sigma \}    \\
    \pderiv[\sigma,\nu]{r + s}{\alpha} &= \pderiv[\sigma,\nu]{r}{\alpha} \cup \pderiv[\sigma,\nu]{s}{\alpha} \\
    \pderiv[\sigma,\nu]{r \cdot s}{\alpha} &= \pderiv[\sigma,\nu]{r}{\alpha}
                                        \cdot (\ApplySubst\sigma s) \cup \{ \SS \mid \Null
                                    (r)\nu, \SS \in \pderiv[\sigma,\nu]{s}{\alpha} \} \\
    \pderiv[\sigma,\nu]{r^*}{\alpha} &= \pderiv[\sigma,\nu]{r}{\alpha} \cdot
                                  (\ApplySubst\sigma r^*)
    \\
    \pderiv[\sigma,\nu]{\mu x.r}{\alpha}
    &= \pderiv[\sigma{[\mu x.r/x]}, \nu{[\Null(r)\nu[\False/x]/x]}]{r}{\alpha}
      \PUSH \SINGLETON\Rempty
    \\ 
    \pderiv[\sigma,\nu]{x}{\alpha} &= \{ \SINGLETON{\ApplySubst\sigma
                                     x} \mid \alpha=\varepsilon \}
  \end{align*}
  \caption{Partial derivatives of $\mu$-regular expressions for
    $\alpha \in \Sigma \cup \{\varepsilon\}$}
  \label{fig:partial-derivatives-v2}
\end{figure}

The derivative for $\mu$-regular expressions has a
\textbf{different type} than for ordinary regular expressions:
A partial derivative is a set of \textbf{non-empty sequences} (i.e.,
stack fragments) of regular expressions. 
The idea is that deriving a recursion operator $\mu x.r$ pushes the current
context on the stack and starts afresh with the derivation of $r$. In
other words, the derivative function for $\mu$-regular expressions has
the same signature as the transition function for a nondeterministic PDA.

To distinguish operations on stacks
from operations on words over $\Sigma$, we write ``$\PUSH$'' (read ``push'') for the concatenation operator on
stacks. We also use this operator for pattern matching parts of a
stack. We write $\EMPTY$ for the empty stack,
$\SINGLETON{r_1, \dots, r_n}$ for a stack with
$n$ elements, and $\RS$ for any stack of expressions.
We extend the concatenation operator for regular expressions to non-empty
stacks by having it operate on the \textbf{last} (bottom) element of a
stack.
\begin{definition}
  Let $(M, (\cdot), \Rempty)$ be a monoid.
  We lift the monoid operation to non-empty stacks $(\cdot) \in M^+ \times M  \to M^+$ 
  for $\overline{a} \in M^*$ and $a, b \in M$ by
  \begin{align*}
    (\overline{a} \PUSH \SINGLETON{a}) \cdot b
    &= (\overline{a} \PUSH \SINGLETON{a \cdot b})
      \text.
  \end{align*}
  We further lift it pointwise to sets $A \subseteq M^+$ to obtain $(\cdot)\in \Power(M^+)
  \times M \to \Power(M^+)$:
  \begin{align*}
    A \cdot b &= \{ \overline{a} \cdot b \mid \overline{a} \in A \}
                \text.
  \end{align*}
\end{definition}
We use this definition for $M= \Reg (\Sigma,X)$ and also extend the
push operation $(:)$ pointwise to sets of stacks.
\begin{align*}
  (\PUSH) &\in \Power(\Reg (\Sigma,X)^+) \times \Reg (\Sigma,X)^+ \to \Power(\Reg (\Sigma,X)^+) \\
  R \PUSH \SS &= \{ \RS  \PUSH \SS \mid \RS \in R \}
\end{align*}
Most of the time, the second argument will be a singleton stack $\SINGLETON{s}$.

Before we discuss the intricacies of the full definition in 
Figure~\ref{fig:partial-derivatives-v2}, let's first consider a naive
extension of the derivative function in
Figure~\ref{fig:partial-derivatives} to $\mu$-regular expressions
and analyze its problems: 
\begin{align*}
  \pderiv{\mu x.r}{a} &= \pderiv{r[\mu x.r/x]}{a} \PUSH \SINGLETON{\Rempty}
                        \tag{naive unrolling: to be revised}
\end{align*}
Taking the derivative of a recursive definition means to apply the derivative to the unrolled
definition. At the same time, we push an empty context on the stack so that the context of the
recursion does not become a direct part of the derivative. This
proposed definition makes sure that the partial derivative
$\pderiv{r}{a}$ is only ever applied to closed expressions $r\in \Reg (\Sigma)$. Hence, the case of a
free recursion variable $x$ would not occur during the computation of $\pderiv{r}{a}$.

\begin{example}\label{example:left-recursion}
The ``naive unrolling'' definition of the partial derivative has a problem. While it can be shown to
be (partially) correct, it is not well-defined for all
arguments. Consider the left-recursive expression $r = \mu x. \Rempty + x \cdot a$,
which is 
equivalent to $a^*$. Computing its partial derivative according to
``naive unrolling'' reveals that it depends on itself, so that
$\pderiv{r}a$ would be undefined.
\begin{align*}
  \underline{\pderiv{r}{a}}
  &= \pderiv{\Rempty + r \cdot a}{a}  \PUSH  \SINGLETON{\Rempty} \\
  &= (\{\} \cup \pderiv{r \cdot a}{a}) \PUSH \SINGLETON{\Rempty} \\
  &= (\pderiv{r}{a} \cdot a \cup \pderiv{a}{a})  \PUSH \SINGLETON{\Rempty} \\
  &= (\underline{\pderiv{r}{a}} \cdot a \cup \{ \SINGLETON\Rempty \}) \PUSH \SINGLETON{\Rempty}
\end{align*}
We remark that the expression $r$
corresponds to a left-recursive grammar, where the naive construction of a top-down parser using the
method of recursive descent also runs into problems
\cite{AhoLamSetiUllman2007}. There would be no problem with the
right-recursive equivalent $r' = \mu x.\Rempty + a\cdot x$ where the
naive unrolling yields $\pderiv{r'}a = \{ [r', \Rempty] \}$.
Indeed, the work by Winter and others \cite{DBLP:conf/calco/WinterBR11} only allows
guarded uses of the recursion operator, which rules out  expressions like
$r$ from the start and which enables them to use the ``naive
unrolling'' definition of the derivative.
\end{example}

For that reason, the derivative must not simply unroll recursions as they are
encountered. Our definition distinguishes between
left-recursive occurrences of a recursion variable, which must not be
unrolled, and guarded occurrences, which can be unrolled
safely. The derivative function remembers deferred unrollings
in a substitution $\sigma$ and applies them only when it is safe.

These observations lead to the signature of the definition of partial derivative in
Figure~\ref{fig:partial-derivatives-v2}. Its type is
\begin{align*}
  \partial &: (\Sigma \cup \{\varepsilon\}) \times (X \to \Reg (\Sigma, X)) \times (X \to
  \bool) \times \Reg (\Sigma, X) \to \Power ( \Reg (\Sigma)^+) 
\end{align*}
and we write it as $\pderiv[\sigma,\nu]{r}{\alpha}$.
It takes a symbol or an empty string $\alpha\in\Sigma \cup
\{\varepsilon\}$ to derive, a substitution $\sigma: X \to \Reg
(\Sigma, X)$ that maps 
free recursion variables to expressions (i.e., their unrollings), a nullability function $\nu:X \to \bool$
that maps free recursion variables to their nullability, and
the regular expression $r \in \Reg (\Sigma, X)$ to derive as arguments and returns the
partial derivatives as a set of non-empty stacks of
expressions. 

Let's examine how the revised definition guarantees well-definedness.
Example~\ref{example:left-recursion} demonstrates that left recursion is the cause for
non-termination of the naive definition. The problem is that the naive definition indiscriminately
substitutes all occurrences of $x$ by its unfolding and propagates the
derivative into the unfolding. However, this substitution is only safe in
guarded positions (i.e., behind at
least one terminal symbol in the unfolding). To avoid substitution in
unguarded positions, the definition in Figure~\ref{fig:partial-derivatives-v2} reifies this
substitution as an additional argument $\sigma$ and 
takes care to only apply it in guarded positions. 

To introduce this recursion, the derived symbol $\alpha$
ranges over $\Sigma\cup\{\varepsilon\}$ in
Figure~\ref{fig:partial-derivatives-v2}. 
For $\alpha=\varepsilon$, the derivative function unfolds one step of
left recursion. 
\begin{example}\label{example:broken-recursion}
  Recall $r = \mu x. \Rempty + x \cdot a$ from Example~\ref{example:left-recursion}. Observe that
  $\Null (r) \emptyset = \Null (\Rempty + x \cdot a) [\False/x] = \True$.
  \begin{align*}
    \pderiv[\emptyset,\emptyset]{r}{a}
    &= (\pderiv[{[r/x],[\True/x]}]{\Rempty + x \cdot a}{a}) \PUSH \SINGLETON{\Rempty} \\
    &= ( \{ \} \cup \pderiv[{[r/x],[\True/x]}]{x \cdot a}{a}) \PUSH \SINGLETON{\Rempty} \\
    &= ( \{ \SINGLETON{\Rempty}\}) \PUSH    \SINGLETON{\Rempty} \\
    &= \{ \SINGLETON{\Rempty, \Rempty}
      \} 
  \end{align*}
  The spontaneous derivative unfolds one level of left recursion.
  \begin{align*}
    \pderiv[\emptyset,\emptyset]{r}{\varepsilon}
    &= (\pderiv[{[r/x],[\True/x]}]{\Rempty + x \cdot a}{\varepsilon}) \PUSH \SINGLETON{\Rempty} \\
    &= ( \{ \} \cup \pderiv[{[r/x],[\True/x]}]{x \cdot a}{\varepsilon}) \PUSH \SINGLETON{\Rempty} \\
    &= ( \{ \SINGLETON{r \cdot a}\}) \PUSH    \SINGLETON{\Rempty} \\
    &= \{ \SINGLETON{r\cdot a, \Rempty}
    \} 
  \end{align*}
\end{example}
Thus, the spontaneous derivative corresponds to
$\varepsilon$-transitions of the PDA 
that is to be constructed.

\section{Correctness}
\label{sec:correctness}

\begin{figure}[t]
  \begin{mathpar}
    \Ruleform{\sigma \vdash w \in r}
    
    \inferrule{}{\sigma \vdash \varepsilon \in \Rempty}

    \inferrule{}{\sigma \vdash a \in a}

    \inferrule{\sigma \vdash w \in r}{\sigma \vdash w \in r+s}

    \inferrule{\sigma \vdash w \in s}{\sigma \vdash w \in r+s}

    \inferrule{\sigma \vdash v \in r \\ \sigma \vdash w \in s}{\sigma
      \vdash vw \in r\cdot s}

    \inferrule{}{\sigma \vdash \varepsilon \in r^*}

    \inferrule{\sigma \vdash v \in r \\ \sigma\vdash w\in r^*}{\sigma
      \vdash vw \in r^*}

    \inferrule{\sigma[\mu x. r/x] \vdash w \in r}{\sigma \vdash w \in
      \mu x.r}

    \inferrule{\sigma \vdash w\in \mu x.r }{
      \sigma[\mu x. r/x]
      \vdash w \in x}
    \\
    \Ruleform{\sigma \vdash w \in \RS}

    \inferrule{}{\sigma \vdash \varepsilon \in []}

    \inferrule{\sigma \vdash v \in r \\ \sigma \vdash w\in \RS}{
      \sigma \vdash vw \in \SINGLETON r \PUSH \RS}
    \\
    \Ruleform{\sigma \vdash w \in R}

    \inferrule{\sigma \vdash w \in \RS \\ \RS \in R}{\sigma \vdash w \in
      R}
  \end{mathpar}
  \caption{Membership in a $\mu$-regular expression, a stack of
    expressions, and a set of stacks}
  \label{fig:inductive-membership-mu-regular}
\end{figure}

To argue about the correctness of our derivative operation, we define
the membership of a word $w\in\Sigma^*$ in the language of an order-respecting
expression $r\in\Reg (\Sigma, X)$ under an order-closed mapping $\sigma: X \to
\Reg (\Sigma,X)$ inductively by the judgment $\sigma \vdash w \in r$ in
Figure~\ref{fig:inductive-membership-mu-regular} along with $\sigma
\vdash w \in \RS$  for an expression stack $\RS$ and  $\sigma
\vdash w \in R$ for a set of such stacks $R \subseteq \Reg (\Sigma,
X)^*$. 
This inductive definition mirrors the previous
fixed point definition of the language of an expression.
\begin{lemma}\label{lemma:inductive-vs-fixpoint}
  For all $r\in \Reg (\Sigma)$ and $w\in\Sigma^*$.
  $\emptyset \vdash w \in r$ iff $w\in \L (r)$.
\end{lemma}

It is straightforward to prove the following derived rule.
\begin{lemma}\label{lemma:membership-subset}
  If $R \subseteq S \subseteq \Reg (\Sigma, X)^*$,
  then $\sigma \vdash w \in R$ implies $\sigma \vdash w \in S$.
\end{lemma}

\begin{lemma}\label{lemma:membership-apply-substitution}
  Let $r\in \Reg (\Sigma, X)$ and $\sigma : X \to \Reg (\Sigma, X)$ be order-respecting.
  If $\sigma \vdash w \in r$, then $\emptyset \vdash w \in \ApplySubst\sigma r$.
\end{lemma}
The derivation closure $\cderiv\RS a$ of a non-empty closed stack of expressions is defined by
the union of the partial derivatives after taking an arbitrary number
of $\varepsilon$-steps. It is our main tool in proving the correctness
of the derivative.
\begin{definition}
  For $a\in \Sigma$, the derivation closure $\cderiv[\sigma,\nu]{r \PUSH \RS} a$
  is inductively defined as the smallest set of stacks such that
  \begin{enumerate}
  \item $\cderiv[\sigma,\nu]{r \PUSH \RS} a \supseteq \pderiv[\sigma,\nu]{r}{a} \PUSH \RS$ and
  \item $\cderiv[\sigma,\nu]{r \PUSH \RS} a \supseteq \bigcup \{ \cderiv[\sigma,\nu]{\SS \PUSH
    \RS}{a} \mid \SS \in \pderiv[\sigma,\nu]{r}{\varepsilon} \}$.
  \end{enumerate}
\end{definition}
\begin{lemma}[Unfolding]\label{lemma:derivation-unfolding}
  Let $r \in \Reg (\Sigma, X)$ an order-respecting expression,
  $\sigma : X \to \Reg (\Sigma, X)$ order-closed with $\sigma (x) =
  \mu x.s_x$,
  $\nu : X \to \bool$ such that $\nu(x) =
  \Null (\ApplySubst \sigma x)\emptyset$.
  \begin{align*}
    \sigma \vdash w \in r & \Leftarrow
    \emptyset \vdash w \in \pderiv[\sigma,\nu]{ r}{\varepsilon}
  \end{align*}
\end{lemma}
\begin{theorem}[Correctness]\label{lemma:derivation-correct}
  Let $r \in \Reg (\Sigma, X)$ an order-respecting expression,
  $\sigma : X \to \Reg (\Sigma, X)$ order-closed with $\sigma (x) =
  \mu x.s_x$,
  $\nu : X \to \bool$ such that $\nu(x) =
  \Null (\ApplySubst \sigma x)\emptyset$.
  \begin{align*}
    \sigma \vdash aw \in r & \Leftrightarrow
    \emptyset \vdash w \in \cderiv[\sigma,\nu]{\SINGLETON r}{a}
  \end{align*}
\end{theorem}
\begin{proof}
  The direction from left to right is proved by induction on  $\sigma
  \vdash aw \in r$. 

  We demonstrate the right-to-left direction. 

  Suppose that
  $\Delta = \emptyset \vdash w \in \cderiv[\sigma,\nu]{\SINGLETON r}a$ and show
  that
  $\sigma \vdash aw \in r$.

  The proof is by induction on the size of the derivation of $\Delta$. 
  Inversion yields that there is some $\RS \in
  \cderiv[\sigma,\nu]{\SINGLETON r}a$ such that $\emptyset \vdash w
  \in \RS$. Now there are two cases.

  \textbf{Case }$w=\varepsilon$ and $\RS=\SINGLETON{}$ so that the
  empty-sequence-rule $\emptyset \vdash \varepsilon \in \SINGLETON{}
  $ applies. But this case cannot happen because $\RS\ne\SINGLETON{}$.

  \textbf{Case }$\emptyset \vdash vw \in \SINGLETON{s}\PUSH\RS$
  because $\emptyset \vdash v \in s$ and $\emptyset \vdash w \in
  \RS$.

  These two cases boil down to $w=w_1\dots w_n$, $\RS = [r_1, \dots, r_n]$, for some $n\ge1$, and
  $\emptyset \vdash w_1\dots w_n \in [r_1, \dots, r_n]$ because
  $\emptyset \vdash w_i \in r_i$.

  We perform an inner induction on $r$.

  \textbf{Case }$\Rnull$, $\Rempty$, $b\ne a$: contradictory because
  $\cderiv[\sigma,\nu]{\SINGLETON r}a = \emptyset$.

  \textbf{Case }$a$: $\cderiv[\sigma,\nu]{\SINGLETON a}a =
  \{\SINGLETON\Rempty\}$ so that $w=\varepsilon$. Clearly, $\sigma
  \vdash a \in a$.

  \textbf{Case }$r+s$: We can show that $\cderiv[\sigma,\nu]{r+s}a =
  \cderiv[\sigma,\nu]{r}a \cup \cderiv[\sigma,\nu]{s}a$. Assuming that
  $\RS \in \cderiv[\sigma,\nu]{r}a$,  induction on
  $r$ yields $\sigma \vdash aw \in r$ and the $+$-rule yields  $\sigma
  \vdash aw \in r+s$. Analogously for $s$.

  \textbf{Case }$r \cdot s$: We can show that
  $\cderiv[\sigma,\nu]{r\cdot s}a =
  \cderiv[\sigma,\nu]{r}a \cdot (\ApplySubst\sigma s) \cup \{ \SS
  \mid \Null (r)\nu, \SS \in \cderiv[\sigma,\nu]{s}a \}$. There are two cases.

  \textbf{Subcase }$\RS \in \cderiv[\sigma,\nu]{r}{a}\cdot
  (\ApplySubst\sigma s)$. Hence, $\RS = [r_1, \dots, r_n \cdot 
  (\ApplySubst\sigma s)]$ so that $w= w_1\dots w_nw_{n+1}$ and $\emptyset
  \vdash w_1 \in r_1$, \dots, $\emptyset \vdash w_n \in r_n$, and
  $\emptyset \vdash w_{n+1} \in(\ApplySubst\sigma s)$.
  Now, $\RS' = [r_1,\dots, r_n] \in \cderiv[\sigma,\nu]{r}{a}$ and thus
  $\emptyset \vdash w_1\dots w_n \in \cderiv[\sigma,\nu]{r}{a}$. By
  induction, $\sigma \vdash aw_1\dots w_n \in r$. Because
  $\ApplySubst\sigma s$ is closed, we also have $\emptyset \vdash
  w_{n+1} \in(\ApplySubst\sigma s)$ and thus by
  Lemma~\ref{lemma:membership-apply-substitution} $\sigma \vdash
  w_{n+1} \in s$. Taken together $\sigma \vdash aw_1\dots w_nw_{n+1}
  \in r\cdot s$.

  \textbf{Subcase }$\Null(r)\nu$ and $\RS \in
  \pderiv[\sigma,\nu]{s}{a}$. Hence, $\sigma \vdash \varepsilon \in
  r$, by induction $\sigma \vdash aw\in s$, and the concatenation rule yields
  $\sigma \vdash aw \in r\cdot s$.

  \textbf{Case }$r^*$. Because $\RS \in \cderiv[\sigma,\nu]{r}{a}
  \cdot(\ApplySubst\sigma {r^*})$, it must be that $\RS = [r_1,\dots,
  r_n\cdot(\ApplySubst\sigma {r^*})]$ and $w=w_1\dots w_nw_{n+1}$ so
  that  $\emptyset
  \vdash w_1 \in r_1$, \dots, $\emptyset \vdash w_n \in r_n$, and
  $\emptyset \vdash w_{n+1} \in(\ApplySubst\sigma {r^*})$. Proceed as in
  the first subcase for concatenation.

  \textbf{Case }$\mu x.r$. As usual, let $\hat\sigma = \sigma{[\mu
    x.r/x]}$ and $\hat\nu = \nu{[\Null (r)\nu[\False/x]/x]}$. Again,
  $\cderiv[\sigma,\nu]{\mu x.r}a = \cderiv[\hat\sigma,\hat\mu]{r}a
  \PUSH \SINGLETON\Rempty$.
  Hence, $\RS = \RS' \PUSH \SINGLETON\Rempty$ for some $\RS' \in
  \cderiv[\hat\sigma,\hat\mu]{r}{a}$ such that $\emptyset \vdash w \in
  \RS'$. Induction yields that $\hat\sigma \vdash aw \in r$ and
  application of the $\mu$-rule yields $\sigma \vdash aw \in \mu x.r$.

  \textbf{Case }$x$. Then $\cderiv[\hat\sigma,\hat\nu]xa =
  \cderiv[\sigma,\nu]{\mu x.r}a$ if $\hat\sigma = \sigma{[\mu
    x.r/x]}$ and $\hat\nu = \nu{[\Null (r)\nu[\False/x]/x]}$.

  Now $\emptyset \vdash w \in \cderiv[\hat\sigma,\hat\nu]xa$ iff
  exists $\RS \in \cderiv[\hat\sigma,\hat\nu]xa =
  \cderiv[\sigma,\nu]{\mu x.r}a$ such that $\emptyset \vdash w \in
  \RS$. But $\cderiv[\sigma,\nu]{\mu x.r}a =
  \cderiv[\hat\sigma,\hat\nu]{r}a \PUSH \SINGLETON\Rempty$ so that
  $\RS = \RS'\PUSH\SINGLETON\Rempty$ and $\emptyset \vdash w \in \RS'$
  with a smaller derivation tree. Thus, induction yields that
  $\hat\sigma \vdash aw \in r$, application of the $\mu$-rule yields
  $\sigma \vdash aw \in \mu x.r$, and application of the variable rule
  yields $\hat\sigma \vdash aw \in x$, as desired.
  \qed
\end{proof}

\section{Finiteness}

In analogy to Antimirov's finite automaton construction, we aim to use
the set of iterated derivatives as a building block for a pushdown
automaton. In our construction, derivatives end up as pushdown symbols
rather than states: the top of the pushdown plays the role of the state.
It remains to prove that this set is finite to obtain a proper PDA.

Our finiteness argument is based on an analysis of the syntactical
form of the derivatives. It turns out that a derivative is, roughly, a
concatenation of a strictly descending sequence of certain subexpressions of
the initial expression. As this ordering is finite, we obtain a finite
bound on the syntactially possible derivatives.

We start with an analysis of the output of $\pderiv[\sigma,\nu]r\alpha$.
The elements in the stack of a partial derivative are vectors of the form $((h\cdot s_1)\cdot s_2)\cdots s_k$ that 
we abbreviate $h\cdot\vec{s}$, where the $s_i$ are
arbitrary expressions and $h$ is either $\Rempty$ or ${\mu x.r}$
where $\mu x.r$ is closed. 

It turns out that the vectors produced by derivation are always
strictly ascending chains in the subterm ordering of the original
expression, say $t$. We first define this ordering, then we define the
structure of these vectors in Definition~\ref{def:t-sorted-vector}.

\begin{definition}
  Let $r \in \Reg (\Sigma)$ be a closed expression. We define the
  \emph{addressing function} $\ADDRESS_r : \nat^* \partialto \Reg (\Sigma, X)$ by
  induction on $r$.
  \begin{align*}
    \ADDRESS_\Rnull  &= \{ (\varepsilon, \Rnull) \} &
    \ADDRESS_{r+s} &= \{ (\varepsilon, r+s) \} \cup 1.\ADDRESS_r \cup 2.\ADDRESS_s \\
    \ADDRESS_\Rempty &= \{ (\varepsilon, \Rempty) \} &
    \ADDRESS_{r\cdot s} &= \{ (\varepsilon, r\cdot s) \} \cup 1.\ADDRESS_r \cup 2.\ADDRESS_s\\
    \ADDRESS_a &= \{ (\varepsilon,a) \}  & \ADDRESS_{r^*} &= \{ (\varepsilon, r^*) \} \cup 1.\ADDRESS_r\\
    \ADDRESS_{x} &= \{ (\varepsilon, x) \} &
    \ADDRESS_{\mu x.r} &= \{ (\varepsilon, \mu x.r) \} \cup 1.\ADDRESS_r 
  \end{align*}
  Here $i.\ADDRESS$ modifies the function $\ADDRESS$ by prepending $i$ to each
  element of $\ADDRESS$'s domain: 
  \begin{align*}
    (i.\ADDRESS) (w) &=
                     \begin{cases}
                       \ADDRESS (w') & w = i w' \text{ and }\ADDRESS
                       (w')\text{ defined} \\
                       \text{undefined} & \text{otherwise.}
                     \end{cases}
  \end{align*}
\end{definition}
It is well known that $\DOM (\ADDRESS_r)$ is prefix-closed and assigns a
unique $w\in\nat^*$ to each occurrence of a subexpression in $r$.
Let $r_1 = \ADDRESS_r (w_1)$ and $r_2 = \ADDRESS_r (w_2)$ be subexpression
occurrences of $r$. 
We say that \emph{$r_1$ occurs before  $r_2$ in $r$} if $w_1
\OccursBefore w_2$ in the lexicographic order on $\nat^*$:
\begin{mathpar}
  \inferrule{}{\varepsilon \OccursBefore w}

  \inferrule{ i < j }{ i v \OccursBefore j w}

  \inferrule{v \OccursBefore w }{ i v \OccursBefore i w }
\end{mathpar}
We write $w_1 \OccursStrictlyBefore w_2$ if $w_1 \OccursBefore w_2$
and $w_1\ne w_2$, in which case we say that $r_1$ occurs strictly
before $r_2$.
\begin{lemma}\label{lemma:no-infinite-chains}
  For each closed expression $r\in \Reg (\Sigma)$, the strict lexicographic
  ordering $\OccursStrictlyBefore$ on $\DOM (\ADDRESS_r)$ has no
  infinite chains.
\end{lemma}
\begin{definition}
  Let $t \in \Reg (\Sigma)$ be a closed expression such that each
  variable occurring in $t$ is bound exactly once. The
  \emph{unfolding substitution $\sigma_t$} is defined by induction on $t$.
  \begin{align*}
    \sigma_\Rnull  & = [] & \sigma_{r+s} & = \sigma_r \cup \sigma_s \\
    \sigma_\Rempty & = [] & \sigma_{r\cdot s} & = \sigma_r \cup \sigma_s\\
    \sigma_a       & = [] & \sigma_{r^*} &= \sigma_r \\
    \sigma_x       & = [] & \sigma_{\mu x.r} &= [\mu x.r / x] \cup \sigma_r
  \end{align*}
\end{definition}
\begin{definition}\label{def:t-sorted-vector}
  A vector $\vec{s} = (s_1\cdot s_2)\cdots s_k$ is \emph{$t$-sorted} if for all
  $1\le i < j \le k$:
    $s_i$ and $s_j$ are subexpressions of $t$ and
    $s_i$ occurs strictly before $s_j$, which means that there are
    $w_1, \dots, w_k \in \nat^*$
    such that $s_i  = {\ADDRESS_t (w_i)}$ and $w_i
    \OccursStrictlyBefore w_{i+1}$, for $1\le i < k$.

    For a $t$-sorted vector $\vec{s} = (s_1\cdot s_2)\cdots s_k$ define two forms of expressions:
    \begin{description}
    \item[top:] $\ApplySubst{\sigma_t}{(\Rempty \cdot \vec{s})}$.
    \item[rec:] $\ApplySubst{\sigma_t}{(( {\mu x.s})\cdot \vec{s})}$
      where $\mu x.s$ is a subexpression of $t$ and either $\mu x.s$ or an occurrence of $x$ is strictly
      before $s_i$, for all $1\le i \le k$.
    \end{description}
    A stack $\RS = [r_1, \dots, r_n]$ (for $n\ge1$) has form
    $\mathbf{top}^+$ if $r_1, \dots, r_n$ have form $\mathbf{top}$.

    A stack $\RS = [r_1, \dots, r_n]$ (for $n\ge1$) has form $\mathbf{rec}.\mathbf{top}^*$ if $r_1$
    has form $\mathbf{rec}$ and $r_2, \dots, r_n$ have form $\mathbf{top}$.
\end{definition}

Next, we show that all derivatives and partial derivatives of
subexpressions of a closed expression $t$ have indeed one of the forms
$\mathbf{top}^+$ or $\mathbf{rec}.\mathbf{top}^*$. 

\begin{lemma}[Classification of derivatives]\label{lemma:classification-of-stack-elements-finally}
  Suppose that $t \in \Reg (\Sigma)$ is a closed expression, $r \in \Reg (\Sigma,X)$ is a subexpression of
  $t$, $\sigma : X \to \Reg (\Sigma,X)$ is order-closed with $\sigma(x) =
  \mu x.s$ (for $x\in X$ and $\mu x.s$ a subterm of $t$), and $\nu : X \to
  \bool$ such that $\nu (x) = \Null (\ApplySubst\sigma x)\emptyset$. 
  If $\RS = [r_1, \dots, r_n]  \in \pderiv[\sigma,\nu]ra$, then $n\ge1$ and $\RS$ 
  has form $\mathbf{top}^+$ and each $r_i = h_i \cdot \vec{s_i}$ for some $t$-sorted $\vec{s_i}$ which is
  before $r$.
\end{lemma}

\begin{lemma}[Classification of spontaneous derivatives]\label{lemma:classification-of-stack-elements-epsilon}
  Suppose that $t \in \Reg (\Sigma)$ is a closed expression, $r \in \Reg (\Sigma,X)$ is a subexpression of
  $t$, $\sigma : X \to \Reg (\Sigma,X)$ is order-closed with $\sigma(x) =
  \mu x.s$ (for $x\in X$ and $\mu x.s$ a subterm of $t$), and $\nu : X \to
  \bool$ such that  $\nu (x) = \Null (\ApplySubst\sigma x)\emptyset$. 
  If $\RS = [r_1, \dots, r_n]  \in \pderiv[\sigma,\nu]r\varepsilon$, then $n\ge1$ and $\RS$ 
  has form $\mathbf{rec}.\mathbf{top}^*$ and each  $r_i = h_i \cdot \vec{s_i}$ for some $t$-sorted $\vec{s_i}$ which is
  before $r$.
\end{lemma}

\begin{lemma}[Classification of derivatives of vectors]\label{lemma:closure-form-finally}
  Let $t \in R (\Sigma)$ be a closed expression and $t_0$ be closed of
  form \textup{\textbf{top}} or form \textup{\textbf{rec}} with respect to
  $t$. Then the elements of $\pderiv[\emptyset,\emptyset]{t_0}a$ are stacks  of the form
  $\mathbf{top}^+$ as in 
  Lemma~\ref{lemma:classification-of-stack-elements-finally} and the
  elements of  $\pderiv[\emptyset,\emptyset]{t_0}\varepsilon$ are
  stacks of the form $\mathbf{rec}.\mathbf{top}^*$.
\end{lemma}


We define the set of iterated partial derivatives as the expressions that
may show up in the stack of a partial derivative. This set will serve
as the basis for defining the set of pushdown symbols of a PDA.

\begin{definition}[Iterated Partial Derivatives]
  Let $t \in \Reg (\Sigma)$ be a closed expression.
  Define $\IPD(t)$, the set of iterated partial derivatives of $t$, as the smallest set such that
  \begin{itemize}
  \item $\Rempty \cdot t \in \IPD (t)$;
  \item if $r \in \IPD (t)$ and  $[t_1,\dots, t_n] \in \pderiv[\emptyset,\emptyset]{r}a$,
    then $t_j \in \IPD(t)$, for all $1\le j\le n$; and
  \item if $r \in \IPD (t)$ and  $[t_1,\dots, t_n] \in \pderiv[\emptyset,\emptyset]{r}\varepsilon$,
    then $t_j \in \IPD(t)$, for all $1\le j\le n$.
  \end{itemize}
\end{definition}
\begin{lemma}[Closure]
  Let $t \in \Reg (\Sigma)$ be a closed expression.
  Then all elements of $\IPD(t)$ either have form \textup{\textbf{top}} or \textup{\textbf{rec}} with respect to $t$. 
\end{lemma}
\begin{proof}
  Follows from
  Lemmas~\ref{lemma:classification-of-stack-elements-finally},
  \ref{lemma:classification-of-stack-elements-epsilon},
  and~\ref{lemma:closure-form-finally}. 
\end{proof}

\begin{lemma}[Finiteness]\label{lemma:finiteness}
  Let $t \in \Reg (\Sigma)$ be closed.
  Then $\IPD (t)$ is finite.
\end{lemma}
\begin{proof}
  By construction, the elements of $\IPD (t)$ are all closed and have
  either form \textbf{top} or form~\textbf{rec}, which is a vector of
  the form $\ApplySubst{\sigma_t}{( h\cdot \vec{s})}$ where $\vec{s}$
  is $t$-sorted. As $t$ is a finite expression and a $t$-sorted vector
  is strictly decreasing, there are only finitely many candidates for
  $\vec{s}$ (by Lemma~\ref{lemma:no-infinite-chains}).

  The head $h$ of the vector is either $\Rempty$ or it is a subexpression of
  $t$ of the form $\mu x.s_x$. Hence, there are only finitely many
  choices for $h$. 

  Thus $\IPD (t)$ is a subset of a finite set and hence finite.
  \qed
\end{proof}
\section{Automaton construction}
\label{sec:automaton-construction}

Given that the derivative for a closed $\mu$-regular expression gives rise to
a finite set of iterated partial derivatives, we use that set as the
pushdown alphabet to construct a nondeterministic pushdown automaton
that recognizes the same language. This construction is
straightforward as its transition function corresponds exactly to the
derivative and the spontaneous derivative function.

\begin{definition}\label{def:unguarded-pda}
  Suppose that $t \in \Reg (\Sigma)$ is closed.
  Define the PDA $\UA (t) = (Q, \Sigma, \Gamma, \delta, q_0, Z_0)$ by
  a singleton set $Q = \{ q \}$, $\Gamma = \IPD (t)$, $q_0 = q$, $Z_0
  = \Rempty\cdot r$, and $\delta \subseteq Q \times
  (\Sigma\cup\{\varepsilon\})\times     \Gamma \times Q \times \Gamma^*$ as the smallest relation such
  that 
  \begin{itemize}
  \item $(q, a, s, q, \SS) \in \delta$ if $\SS \in \pderiv[\emptyset,\emptyset]{s}{a}$, 
    for all $s \in \Gamma$, $\SS\in\Gamma^*$, $a\in\Sigma$;
  \item $(q, \varepsilon, s, q, {\SS}) \in
    \delta$ if $
    \SS \in \pderiv[\emptyset,\emptyset]{s}{\varepsilon}
    $,  
    for all $s \in \Gamma$, $\SS\in\Gamma^*$;
  \item $(q, \varepsilon, s, q, \varepsilon)$, for all $s\in \Gamma$ with $\Null (s)\emptyset$.
  \end{itemize}
\end{definition}

\begin{theorem}[Automaton correctness]\label{lemma:automaton-correctness}
  For all closed expressions $t\in R (\Sigma)$, $ \L (t) = \L(\UA (t))$.
\end{theorem}
\begin{proof}
  Let $\UA (t) = (Q, \Sigma, \Gamma, \delta, q_0, Z_0)$.
  We prove a generalized statement from which the original statement
  follows trivially: for all $\RS \in \IPD (t)^*$, $\emptyset \vdash w \in
  \RS$ iff $(q, \RS, w) \StepsTo^* (q, \varepsilon, \varepsilon)$.
  The proof in the left-to-right direction is by induction on the derivation of $\emptyset \vdash w
  \in \RS$.

  \textbf{Case }$\emptyset \vdash \varepsilon \in
  \SINGLETON{}$. Immediate.

  \textbf{Case }$\emptyset \vdash w \in  \RS$ because
  $w=w_1w_2$, $\RS = \SINGLETON r\PUSH\RS'$, $\emptyset \vdash w_1 \in
  r$, and $\emptyset \vdash w_2 \in \RS'$.
  By induction, we find that $(q, \SINGLETON r, w_1) \vdash^* (q,
  \SINGLETON{}, \varepsilon)$. By a standard argument that means  $(q, \SINGLETON r\PUSH\RS', w_1w_2) \vdash^* (q,
  \RS', w_2)$. By the second inductive hypothesis, we find that $ (q,
  \RS', w_2) \vdash^* (q, \SINGLETON{}, \varepsilon)$. Taken together,
  we obtain the desired result.

  Now we consider the derivation of $\emptyset \vdash w \in r$ by
  performing a case analysis on $w$ and using Lemma~\ref{lemma:derivation-correct}.

  \textbf{Case }$\varepsilon$. In this case,
  $\emptyset \vdash \varepsilon \in r$ iff $\Null
  (r_i)\emptyset$ iff $(q, \varepsilon, r, q, \varepsilon)
  \in \delta$ so that  $(q, \SINGLETON{r}, \varepsilon) \StepsTo^+ (q, \varepsilon,
  \varepsilon)$.

  \textbf{Case }$aw$.
  In this case $\emptyset \vdash aw \in r$.
  By Lemma~\ref{lemma:derivation-correct}, $\emptyset \vdash aw \in
  r$ is equivalent to $\emptyset \vdash w \in
  \cderiv[\emptyset,\emptyset]{\SINGLETON{r}} a$ and we perform a
  subsidiary induction on its definition.
  That is, either 
  $\exists\SS\in\pderiv[\emptyset,\emptyset]{r} a$ such that
  $\emptyset \vdash w \in \SS$. In that case, $\UA (t)$ has a
  transition  $(q, r, aw) \StepsTo (q, \SS,
  w)$ by definition of $\delta$. By induction we know that $(q, \SS,
  w) \StepsTo^+ (q, \varepsilon, \varepsilon)$. 

  Alternatively, $\exists\SS\in\pderiv[\emptyset,\emptyset]{r}
  \varepsilon$ such that $\emptyset \vdash aw \in \SS$. In this
  case, $(q, r, aw) \vdash (q, \SS, aw)$ is a transition and by
  induction we have $(q, \SS, aw) \vdash^+ (q, \varepsilon,
  \varepsilon)$.

  \textbf{Right-to-left direction}. By induction on the
  length of $(q, \RS, w) \vdash^*   (q, \SINGLETON{}, \varepsilon)$.

  \textbf{Case }length $0$: it must be $\RS=\SINGLETON{}$ and
  $w=\varepsilon$. Obviously, $\emptyset \vdash \varepsilon \in
  \SINGLETON{}$.

  \textbf{Case }length $>0$: Thus the first configuration must have
  the form $(q, \SINGLETON{s}\PUSH\RS, w)$. There are three
  possibilities.

  \textbf{Subcase }$(q, \SINGLETON{s}\PUSH\RS, w) \vdash (q, \SS\PUSH
  \RS, w')$ if $w=aw'$ and $\SS \in \pderiv s a$. We split the run of
  the automaton at the point where $\SS$ is first consumed: let $w' = w_1w_2$
  such that $(q, \SS\PUSH\RS, w_1w_2) \vdash^* (q, \RS, w_2) \vdash^*
  (q, \SINGLETON{}, \varepsilon)$.
  Hence, there is also a shorter run on $w_1$: $(q, \SS, w_1) \vdash^* (q,
  \SINGLETON{}, \varepsilon)$. Induction yields $\emptyset \vdash w_1
  \in \SS$. By Lemma~\ref{lemma:derivation-correct}, we also have a
  derivation $\emptyset \vdash aw_1 \in s$. By induction on the $\RS$
  run, we obtain $\emptyset \vdash w_2 \in \RS$ and applying the stack
  rule yields $\emptyset \vdash aw_1w_2 \in \SINGLETON{s}\PUSH\RS$ or
  in other words $\emptyset \vdash w \in \SINGLETON{s}\PUSH\RS$.

  \textbf{Subcase }$(q, \SINGLETON{s}\PUSH\RS, w) \vdash (q, \SS\PUSH
  \RS, w)$ if $\SS \in \pderiv s \varepsilon$. We split the run of the
  automaton at the point where $\SS$ is first consumed: let $w' =
  w_1w_2$ such that  $(q, \SS\PUSH\RS, w_1w_2) \vdash^* (q, \RS, w_2) \vdash^*
  (q, \SINGLETON{}, \varepsilon)$. Hence there is also a shorter run
  on $w_1$: $(q, \SS, w_1) \vdash^* (q, \SINGLETON{},
  \varepsilon)$. By induction, we have a
  derivation $\emptyset \vdash w_1 \in \SS$, which yields $\emptyset
  \vdash w_1\in s$ by Lemma~\ref{lemma:derivation-unfolding}, and a derivation
  $\emptyset \vdash w_2 \in \RS$, which we can combine to $\emptyset \vdash
  w_1w_2 \in \SINGLETON{s}\PUSH \RS$ as desired.

  \textbf{Subcase }$(q, \SINGLETON{s}\PUSH\RS, w) \vdash (q, 
  \RS, w)$ if $\Null (s)\emptyset$. By induction, $\emptyset \vdash w
  \in \RS$. As $\Null (s)\emptyset$, it must be that $\emptyset \vdash
  \varepsilon \in s$. Hence, $\emptyset \vdash w\in \SINGLETON{s}\PUSH
  \RS$.
  \qed
\end{proof}

If all recursion operators in an expression $t$ are guarded, in the
sense that they consume some input before entering a recursive call,
then all $\varepsilon$-transitions in the constructed automaton pop
the stack. In fact, when restricting to guarded expressions, the
spontaneous derivative function is not needed at all, which explains
the simplicity of the derivative in the work of Winter and coworkers \cite{DBLP:conf/calco/WinterBR11}.

\subsubsection*{Acknowledgments.} The thoughtful comments of the
anonymous reviewers helped improve the presentation of this paper.
\clearpage

\bibliography{main}
\clearpage
\appendix{}
\section{APPENDIX WITH SELECTED PROOFS}
\begin{proof}[of Lemma~\ref{lemma:lang-is-monotone}]
  Let $L\subseteq L'$, $\eta_1 = \eta[x \mapsto L]$,  $\eta_2 = \eta[x \mapsto L']$, and proceed by induction on $r$.

  \textbf{Cases }$\Rnull$, $\Rempty$, $a$: immediate.

  \textbf{Cases }$r\cdot s$, $r + s$, $r^*$: Immediate by induction.

  \textbf{Case }$x$: Immediate because $\eta_1 (x) \subseteq \eta_2
  (x)$ by assumption.

  \textbf{Case }$y\ne x$: Immediate because $\eta_1 (y) = \eta_2 (y)$.

  \textbf{Case }$\mu y.r$: Requires an auxiliary fixed point induction to
  prove containment.
  \begin{align*}
    \Lang (\mu y.r, {\eta_1})
    &= \LFP\ L. \Lang (r, {\eta_1[y \mapsto L]}) 
    &&\subseteq \LFP\ L. \Lang (r, { \eta_2[y \mapsto L]})
    &&= \Lang (\mu y.r, {\eta_2})
  \end{align*}
\end{proof}
\begin{proof}[of Lemma~\ref{lemma:epsilon-in-empty-l}]
  The proof is by induction on $r$.

  \textbf{Case }$\Rnull$: obvious.

  \textbf{Case }$\Rempty$:  immediate.

  \textbf{Case }$a\in \Sigma$: obvious.

  \textbf{Case }$r+s$: immediate by induction.

  \textbf{Case }$r \cdot s$: if $\varepsilon \notin\Lang (r \cdot s, { \eta[x \mapsto
  \emptyset]})$, then  $\varepsilon \notin\Lang (r, { \eta[x \mapsto
  \emptyset]})$ or  $\varepsilon \notin\Lang (s, { \eta[x \mapsto
  \emptyset]})$. In either case, the result is immediate by the
  inductive hypothesis on $r$ or $s$, respectively.

  \textbf{Case }$r^*$: contradicts assumption.

  \textbf{Case }$x$: by the assumption on $L$.

  \textbf{Case }$y \ne x$: immediate.

  \textbf{Case }$\mu x.r$: immediate because $\Lang (\mu x.r, {\eta})$
  is independent of $\eta(x)$.

  \textbf{Case }$\mu y.r$ for $y\ne x$: we need to show that
  ``$\varepsilon \notin\Lang (\mu y.r, { \eta[x \mapsto
  \emptyset]}) $ implies that, for all $L$,  $\varepsilon \notin L$ implies
  $\varepsilon \notin \Lang (\mu y. r, { \eta[x \mapsto
  L]})$.''

  By definition of $\Lang$, this statement is equivalent to
  ``$\varepsilon \notin\LFP\ Y.\Lang (r, { \eta[x \mapsto
  \emptyset][y \mapsto Y]}) $ implies that, for all $L$,  $\varepsilon \notin L$ implies
  $\varepsilon \notin \LFP\ Y.\Lang ( r, { \eta[x \mapsto
  L][y \mapsto Y]})$.''

  By unrolling the fixed point, we obtain a further equivalent statment
  ``$\varepsilon \notin \Lang (r, { \eta[x \mapsto
  \emptyset][y \mapsto \LFP\ Y.\Lang (r, \eta[x \mapsto
  \emptyset])]}) $ implies that, for all $L$,  $\varepsilon \notin L$ implies
  $\varepsilon \notin \Lang (r, {\eta[x \mapsto
  L][y \mapsto \LFP\ Y.\Lang (r, \eta[x \mapsto
  L])]}) $.''

  By this statement holds by induction with
  $\eta [y \mapsto \LFP\ Y.\Lang (r, { \eta[x \mapsto  L]})]$
  substituted for $\eta$.
  \qed
\end{proof}
\begin{proof}[of Lemma~\ref{lemma:nullability-fixpoint}]
  By Tarksi's theorem, we can rewrite the fixed point using an auxiliary
  function $F_\Null (b) = \Null (r) \nu[x \mapsto b]$:
  \begin{align*}
    \LFP\ b.\Null (r)\nu[x \mapsto b]
    &= \bigvee_{i \in \nat} F_\Null^{(i)} (\False)
    \\
    &= \False \vee F_\Null (\False) \vee
    \bigvee_{i\in\nat}F_\Null^{(i+2)} (\False)
    \\
    &= F_\Null (\False) \vee
    \bigvee_{i\in\nat}F_\Null^{(i)} (F_\Null (\False))
  \end{align*}
  There are two cases. If $F_\Null (\False) = \False$, then
  $F_\Null^{(i)} (F_\Null (\False)) = F_\Null^{(i)} (\False) = \False$
  for all $i\in \nat$. Hence, $\bigvee_{i \in \nat} F_\Null^{(i)}
  (\False) = \False$.

  If $F_\Null (\False) = \True$, then
  $F_\Null^{(i)} (F_\Null (\False)) = F_\Null^{(i)} (\True) = \True$
  for all $i\in \nat$ because $F_\Null$ is a monotone function.
  Hence, $\bigvee_{i \in \nat} F_\Null^{(i)}
  (\False) = \True$.
  \qed
\end{proof}
\begin{proof}[of Lemma~\ref{lemma:correctness-of-null}]
  We proceed by induction on $r$.

  \textbf{Case }$\Rnull$, $\Rempty$, $a\in\Sigma$: obvious.

  \textbf{Case }$r + s$, $r \cdot s$: Immediate by induction.

  \textbf{Case }$r^*$: obvious.
  
  \textbf{Case }$x$: immediate by assumption on $\eta$ and $\nu$.

  \textbf{Case }$\mu x.r$:
  Suppose that $\varepsilon \in \Lang (\mu  x.r, {\eta})$.
  By definition and fixed point unrolling, we have
  \begin{align*}
    \Lang (\mu  x.r, {\eta})
    &= \LFP\ L. \Lang (r, { \eta[x \mapsto L]})
    &&= \Lang (r, { \eta[x \mapsto \LFP\ L. \Lang (r, \eta[x \mapsto L])]})
  \end{align*}
  Now we can argue as follows
  \begin{align*}
    & \varepsilon \in \Lang (r, { \eta[x \mapsto \LFP\ L. \Lang (r
    \eta[x \mapsto L])]})
    \\
    \Leftrightarrow & \text{ by Lemma~\ref{lemma:epsilon-in-empty-l}}
    \\
    & \varepsilon \in \Lang (r, { \eta[x \mapsto \emptyset]})
    \\
    \Leftrightarrow & \text{ by induction because $\eta[x \mapsto  \emptyset] \models \nu[x \mapsto \False]$}
    \\
    & \Null (r)\nu[x \mapsto\False]
    \\
    \Leftrightarrow & \text{ by
      Lemma~\ref{lemma:nullability-fixpoint}}
    \\
    & \LFP\ b.\Null (r)\nu[x\mapsto b] = \Null (\mu x.r)\nu
  \end{align*}
  \qed
\end{proof}

\clearpage{}
\begin{proof}[of Lemma~\ref{lemma:inductive-vs-fixpoint}]
    For all $r\in \Reg (\Sigma)$, $\emptyset \vdash w \in r$ iff $w
    \in \L (r)$.

    \textbf{Left to right}: Perform an induction on $\emptyset \vdash
    w \in r$.

    We need to generalize to, for all $r \in \Reg (\Sigma, X)$ and
    order-respecting $\sigma : X \to \Reg (\Sigma, X)$,
    $\sigma \vdash w \in r$ implies $w \in \L[{\eta}] (r)$ where $\eta (x)
    = \L[{ (\eta\setminus x)}] (\sigma (x))$. The right hand side is
    well-defined because $\sigma$ is order-respecting.

    \textbf{Case }$\sigma \vdash \varepsilon \in \Rempty$: clearly
    $\varepsilon \in \{\varepsilon\} = \L[{\eta}] (\Rempty)$ for any
    $\eta$.

    \textbf{Case }$\sigma \vdash a \in a$: clearly $a \in \{ a \} = \L[{\eta}]
    (a)$ for any $\eta$.

    \textbf{Case }$\sigma \vdash w \in r+s$ because $\sigma \vdash w
    \in r$. By induction $w \in \L[{ \eta}] (r) \subseteq L (r) \eta \cup
    \L[{\eta}] (s) = \L (r+s, { \eta })$.

    \textbf{Case }$\sigma \vdash w \in r+s$ because $\sigma \vdash w
    \in s$. Analogous to previous case.

    \textbf{Case }$\sigma \vdash vw \in r\cdot  s$ because $\sigma
    \vdash v \in r$ and $\sigma \vdash w \in s$. By induction $v \in
    \L[{\eta}] (r)$ and $w\in \L[{\eta}] (s)$, hence $vw \in \L[{\eta}] (r) \cdot
    \L[{\eta}] (s) = \L[{\eta}] (r \cdot s)$.

    \textbf{Case }$\sigma \vdash w \in \mu x.r$ because $\sigma[\mu
    x.r/x] \vdash  w \in r$. By induction $w \in \L (r, { \eta[\L (\mu
    x.r, \eta)/x]}) = \L (\mu x.r, {\eta})$. 

    \textbf{Case }$\sigma[\mu    x.r/x] \vdash w \in x$ because
    $\sigma \vdash w \in \mu x .r$. By induction, we obtain that $w
    \in \L[{ \eta}] (\mu x.r) = \L (x, {\eta[\L (\mu x.r, \eta)/x]})$ which
    fits with $\sigma[\mu x.r/x]$.

    \textbf{Right to left}. Perform an induction on $r$ to show that
    for all $r \in \Reg (\Sigma, X)$ and order-respecting $\sigma : X
    \to \Reg (\Sigma, X)$ such that $\eta (X) \subseteq \L[{ (\eta \setminus x)}] (\sigma
    (x))$, 
    $w \in \L[{ \eta}] (r)$ implies $\sigma \vdash w \in r$.

    \textbf{Case }$\Rnull$. Void.

    \textbf{Case }$\Rempty$. $w\in \L[{\eta}] (\Rempty)$ implies
    $w=\varepsilon$ and $\sigma \vdash \varepsilon \in \Rempty$ for
    all $\sigma$.

    \textbf{Case }$a$. $w \in \L[{\eta}] (a)$ implies $w=a$ and $\sigma
    \vdash a \in a$ for all $\sigma$.

    \textbf{Case }$r+s$. $w \in \L[{\eta}] (r+s) $ implies $w \in \L[{\eta}]
    (r)$ or $w\in\L[{\eta}] (s)$. By induction $\sigma \vdash w \in r$
    or $\sigma \vdash w \in s$. In both cases, we can conclude $\sigma
    \vdash w \in r+s$ for all $\sigma$.

    \textbf{Case }$r \cdot s $. $w\in\L[{\eta}] (r\cdot s)$ implies $w =
    uv$ and $u \in \L[{\eta}] (r)$ and $v \in \L[{\eta}] (s)$. By induction
    $\sigma \vdash u \in r$ and $\sigma \vdash v \in s$ and hence
    $\sigma \vdash uv \in r \cdot s$.

    \textbf{Case }$\mu x.r$. $w \in \L[{ \eta}] (\mu x.r) = \LFP\lambda
    V. \L[{\eta[V/x]}] (r)$. We proceed by fixed point induction. Assuming
    that for all $v \in V = \L[{\eta[V/x]}] (x)$, $\sigma[\mu x.r/x] \vdash v \in x$, we
    obtain by induction that $\sigma[\mu x.r/x] \vdash w \in r$ and
    hence $\sigma \vdash w \in \mu x .r $.

    \textbf{Case }$x$. $w \in \L[{\eta}] (x) = \eta (x) \subseteq \L[{ (\eta \setminus x)}]
    (\sigma (x))$. By the assumption of the
    fixed point induction, we find that $\sigma[\mu x.r/x] \vdash w \in
    x$.
    \qed
\end{proof}

\clearpage{}
\begin{proof}[of Lemma~\ref{lemma:membership-apply-substitution}]
  By induction on the derivation of  $\sigma \vdash w \in r$.

  \textbf{Rule }$\inferrule{}{\sigma \vdash \varepsilon \in
    \Rempty}$. Immediate.

  \textbf{Rule }$\inferrule{}{\sigma \vdash a \in a}$. Immediate.

  \textbf{Rule }$\inferrule{\sigma \vdash w \in r}{\sigma \vdash w \in
    r+s}$. Immediate by application of the inductive hypothesis and
  $\ApplySubst\sigma{(r+s)} = (\ApplySubst\sigma r) + (\ApplySubst\sigma
  s)$.

  \textbf{Rule }$\inferrule{\sigma \vdash w \in s}{\sigma \vdash w \in
    r+s}$. Analogous to previous.

  \textbf{Rule }$\inferrule{\sigma \vdash v \in r \\ \sigma \vdash w \in s}{\sigma
    \vdash vw \in r\cdot s}$. Immediate by application of the inductive hypothesis and
  $\ApplySubst\sigma{(r\cdot s)} = (\ApplySubst\sigma r) \cdot (\ApplySubst\sigma
  s)$.

  \textbf{Rule }$\inferrule{}{\sigma \vdash \varepsilon \in
    r^*}$. Immediate.

  \textbf{Rule }$\inferrule{\sigma \vdash v \in r \\ \sigma\vdash w\in r^*}{\sigma
    \vdash vw \in r^*}$. Immediate by application of the inductive hypothesis and
  $\ApplySubst\sigma{(r^*)} = (\ApplySubst\sigma r)^*$.

  \textbf{Rule }$\inferrule{\sigma[\mu x. r/x] \vdash w \in r}{\sigma \vdash w \in
    \mu x.r}$.

  Because $\mu x.r \in \Reg (\Sigma, X)$ is order-respecting, we can assume
  that $x\notin X$ and  $y \prec x$, for all $y\in \FV (\mu x.r) \subseteq X = \DOM
  (\sigma)$.
  Hence, $\sigma[\mu x.r/x] : (X\cup\{x\}) \to \Reg (\Sigma,
  X\cup\{x\})$ is order-closed and $r \in \Reg (\Sigma,
  X\cup\{x\})$ is order-respecting. Thus, induction is applicable and  yields $\emptyset
  \vdash w \in \ApplySubst{\sigma[\mu x.r/x]}{r}$. By fixed point
  folding we obtain  $\emptyset
  \vdash w \in \ApplySubst{\sigma}{\mu x.r}$.

  \textbf{Rule }$\inferrule{\sigma \vdash w\in \mu x.r }{
      \sigma[\mu x. r/x]
      \vdash w \in x}$.

  By induction, $\emptyset \vdash w \in \ApplySubst\sigma {(\mu
    x.r)}$. Clearly, $\ApplySubst\sigma {(\mu x.r)} =
  \ApplySubst{\sigma[\mu x. r/x]}{x}$, so that $\emptyset \vdash w \in
  \ApplySubst{\sigma[\mu x. r/x]} x$.
  \qed
\end{proof}

\begin{proof}[of Theorem~\ref{lemma:derivation-correct}]
  Prove the direction from left to right by induction on $\sigma
  \vdash aw \in r$. 

  \textbf{Rule }$\inferrule{}{\sigma \vdash \varepsilon \in \Rempty}$:
  contradictory as $\varepsilon$ cannot be written in the form $aw$.

  \textbf{Rule }$\inferrule{}{\sigma \vdash a \in a}$: in this case
  $w=\varepsilon$ and $\Rempty\in \cderiv[\sigma,\nu]{a}{a}$. Hence, $\emptyset
  \vdash \varepsilon \in  \cderiv[\sigma,\nu]{a}{a}$ is 
  derivable.

  \textbf{Rule }$\inferrule{\sigma \vdash w \in r}{\sigma \vdash w \in
    r+s}$.
  Thus $\sigma \vdash aw \in r+s$ because $\sigma \vdash aw \in r$. By
  induction, $\emptyset \vdash w \in \cderiv[\sigma,\nu]{\SINGLETON r}{a}$ and
  hence $\emptyset \vdash w \in \cderiv[\sigma, \nu]{r+s}{a}$ is also
  provable by Lemma~\ref{lemma:membership-subset} because
  $\pderiv[\sigma, \nu]{r}{\alpha} \subseteq \pderiv[\sigma,
  \nu]{r+s}{\alpha}$. 

  \textbf{Rule }$\inferrule{\sigma \vdash w \in s}{\sigma \vdash w \in
    r+s}$. Analogous.

  \textbf{Rule }$\inferrule{\sigma \vdash v \in r \\ \sigma \vdash w \in s}{\sigma
    \vdash vw \in r\cdot s}$. There are two cases.

  \textbf{Subcase }$\sigma \vdash aw \in r\cdot s$ because $w=w_1w_2$
  and $\sigma \vdash aw_1 \in r$ and $\sigma \vdash w_2 \in s$.
  By induction, $\emptyset \vdash w_1 \in \cderiv[\sigma, \nu]{r}{a}$
  which means (inversion) there is some $\RS\in \cderiv[\sigma, \nu]{r}{a}$ such
  that $\emptyset \vdash w_1 \in \RS$. By Lemma~\ref{lemma:membership-apply-substitution},
  $\sigma \vdash w_2 \in s$ iff  $\emptyset \vdash w_2 \in
  \ApplySubst\sigma  s$.
  By the concatenation rule we obtain  $\emptyset \vdash w_1w_2 \in
  \RS \cdot ({\ApplySubst\sigma s})$.
  Now
  $\RS \cdot (\ApplySubst\sigma s) \in \cderiv[\sigma, \nu]{r}{a} \cdot (\ApplySubst\sigma
  s) \subseteq \cderiv[\sigma, \nu]{r\cdot s}{a}$ so that
  $\emptyset \vdash w_1w_2 \in \cderiv[\sigma, \nu]{r\cdot s}{a}$.

  \textbf{Subcase }$\sigma \vdash aw \in r\cdot s$ because $\sigma
  \vdash \varepsilon \in r$ and $\sigma \vdash aw \in s$. Immediate by
  observing that $\Null (r)\nu = \True$,
  induction on the proof of  $\sigma \vdash aw \in s$, and applying
  Lemma~\ref{lemma:membership-subset}.

  \textbf{Rule }$\inferrule{}{\sigma \vdash \varepsilon \in r^*}$:
  contradictory.

  \textbf{Rule }$\inferrule{\sigma \vdash v \in r \\ \sigma\vdash w\in r^*}{\sigma
    \vdash vw \in r^*}$.  There are two cases.

  \textbf{Subcase }$\sigma \vdash aw \in r^*$ because $w=w_1w_2$ and
  $\sigma \vdash aw_1 \in r$ and $\sigma \vdash w_2 \in r^*$. By
  analogous argumentation as in the first subcase for concatenation,
  we find that $\emptyset \vdash w_1w_2 \in \cderiv[\sigma,
  \nu]{r^*}{a}$.

  \textbf{Subcase }$\sigma \vdash  aw \in r^*$ because $\sigma \vdash
  \varepsilon \in r$ and $\sigma \vdash aw \in r^*$. Immediate by the
  inductive hypothesis for  $\sigma \vdash aw \in r^*$.

  \textbf{Rule }$\inferrule{\sigma[\mu x. r/x] \vdash w \in r}{\sigma \vdash w \in
    \mu x.r}$.
  Let $\hat\sigma = \sigma[\mu x. r/x]$
  and $\hat\nu = \nu[\Null (r)\nu[ \False/x]/x]$.

  Given the assumption $\hat\sigma \vdash aw \in r$, we obtain
  by induction that there exists some $\RS\in \cderiv[\hat\sigma,
  \hat\nu]{r}{a}$ so that $\emptyset \vdash w \in \RS$. Clearly,
  $\emptyset \vdash w \in \RS \PUSH \Rempty$  and thus  $\emptyset \vdash w
  \in \cderiv[\sigma,  \nu]{\mu x.r}{a}$.

  \textbf{Rule }$\inferrule{\sigma \vdash w\in \mu x.r }{ \sigma[\mu x. r/x]
    \vdash w \in x}$.   Let $\hat\sigma = \sigma[\mu x. r/x]$, $\hat\nu =
  \nu[\Null (r)\nu[ \False/x]/x]$.

  Thus, $\sigma[\mu x. r/x] \vdash aw \in x$ because $\sigma \vdash aw
  \in \mu x.r$. 
  By induction, it must be that $\emptyset \vdash w \in
  \cderiv[\sigma,\nu]{\mu x.r}{a}$.
  We find that also $\emptyset \vdash w \in
  \cderiv[\hat\sigma,\hat\nu]{x}{a}$ because
  \begin{align*}
    \cderiv[\hat\sigma,\hat\nu]{ x}{a}
    &=
    \pderiv[\hat\sigma,\hat\nu]{ x}{a} \cup \bigcup \{
    \cderiv[\hat\sigma,\hat\nu]{\SS}{a} \mid \SS \in
    \pderiv[\hat\sigma,\hat\nu]{x}{\varepsilon}\}
    \\&=
    \emptyset \cup \bigcup \{
    \cderiv[\hat\sigma,\hat\nu]{\SS}{a} \mid \SS \in
    \SINGLETON{\mu x.r}\}
    \\&=
    \cderiv[\sigma,\nu]{\mu x.r}{a}
    \text.
  \end{align*}
  \qed
\end{proof}

\begin{proof}[of Lemma~\ref{lemma:classification-of-stack-elements-finally}]
  By induction on $r$. Observe that all substitutions that arise in
  the inductive computation of $\pderiv{r}a$ are subsumed by $\sigma_t$.

  \textbf{Case }$\Rnull$, $\Rempty$: void.

  \textbf{Case }$a$: $\pderiv[\sigma,\nu]aa = \{ \SINGLETON\Rempty \}$
  has form \textbf{top} for empty $\vec s$.

  \textbf{Case }$r+s$: immediate by induction.

  \textbf{Case }$r\cdot s$: If $\SS \in \pderiv[\sigma,\nu]{r\cdot
    s}a$ because $\SS = \RS \cdot (\ApplySubst{\sigma_t} s) $, for some
  $\RS\in \pderiv[\sigma,\nu]ra$, then $\RS$ has the required shape by
  induction.
  The last expression in $\RS$  is either $\Rempty$ 
  or it ends with a vector of $\sigma_t$ substitution instances of
  elements before $r$. As $r$ is before $s$, the final concatenation
  preserves $t$-sortedness and thus $\SS$ has the required shape.
  If $\SS \in
  \pderiv[\sigma,\nu]sa$, then all forms are 
  preserved by induction.

  \textbf{Case }$r^*$:  If $\SS \in \pderiv[\sigma,\nu]{r^*}a$, then
  $\SS = \RS \cdot (\ApplySubst{\sigma_t} r^*)$ for some $\RS\in
  \pderiv[\sigma,\nu]ra$. By induction, $\RS$ has the required
  shape. The last expression in $\RS$ is either $\Rempty$ 
  or it ends with a vector of $\sigma_t$ substitution instances of
  elements before $r$. Thus, the final concatenation preserves
  $t$-sortedness and hence $\SS$ has the required shape.

  \textbf{Case }$\mu x.r$: 
  If $\SS \in \pderiv[\sigma,\nu]{\mu x.r}a$,
  then $\SS = \RS \PUSH \SINGLETON\Rempty$, By induction, $\SS$ has
  form $\mathbf{top}^+$ and $\Rempty$ has form $\mathbf{top}$ with an
  empty vector, so that $\RS$ has form $\mathbf{top}^+$. 

  \textbf{Case }$x$: void.
  \qed
\end{proof}
\begin{proof}[of Lemma~\ref{lemma:classification-of-stack-elements-epsilon}]
  By induction on $r$. Observe that all substitutions that arise in
  the inductive computation of $\pderiv{r}a$ are subsumed by $\sigma_t$.

  \textbf{Case }$\Rnull$, $\Rempty$, $a$: void.

  \textbf{Case }$r+s$: immediate by induction.

  \textbf{Case }$r\cdot s$: If $\SS \in \pderiv[\sigma,\nu]{r\cdot
    s}\varepsilon$ because $\SS = \RS \cdot (\ApplySubst{\sigma_t} s)$, for some
  $\RS\in \pderiv[\sigma,\nu]r\varepsilon$, then $\RS$ has the required shape by
  induction.
  The last expression in $\RS$  is either $\Rempty$ 
  or it ends with a vector of $\sigma_t$ substitution instances of
  elements before $r$. As $r$ is before $s$, the final concatenation
  preserves $t$-sortedness and thus $\SS$ has the required shape.
  If $\SS \in
  \pderiv[\sigma,\nu]s\varepsilon$, then all forms are 
  preserved by induction.

  \textbf{Case }$r^*$:  If $\SS \in \pderiv[\sigma,\nu]{r^*}\varepsilon$, then
  $\SS = \RS \cdot (\ApplySubst{\sigma_t} r^*)$ for some $\RS\in
  \pderiv[\sigma,\nu]r\varepsilon$. By induction, $\RS$ has the required
  shape. The last expression in $\RS$ is either $\Rempty$ 
  or it ends with a vector of $\sigma_t$ substitution instances of
  elements before $r$. Thus, the final concatenation preserves
  $t$-sortedness and hence $\SS$ has the required shape.

  \textbf{Case }$\mu x.r$: 
  If $\SS \in \pderiv[\sigma,\nu]{\mu x.r}\varepsilon$,
  then $\SS = \RS \PUSH \SINGLETON\Rempty$, By induction, $\SS$ has
  form $\mathbf{rec}.\mathbf{top}^*$ and $\Rempty$ has form $\mathbf{top}$ with an
  empty vector, so that $\RS$ has form $\mathbf{rec}.\mathbf{top}^*$. 

  \textbf{Case }$x$: $\pderiv[\sigma,\nu]{x}\varepsilon = \{
  \SINGLETON{\ApplySubst\sigma x} \}$, which has shape $\mathbf{rec}$
  and thus $\mathbf{rec}.\mathbf{top}^*$.
  \qed
\end{proof}
\begin{proof}[of Lemma~\ref{lemma:closure-form-finally}]
  \textbf{Form~top}:
  As $t_0$ is closed, it must be that for a $t$-sorted vector $s_1\cdots s_k$
  \begin{align*}
    \pderiv[\emptyset,\emptyset]{t_0}a
    &= \pderiv[\emptyset,\emptyset]{\ApplySubst{\sigma_t}{(\Rempty \cdot (s_1\cdots s_k))}}a \\
    &= \pderiv[\emptyset,\emptyset]{\ApplySubst{\sigma_t}{(s_1\cdots
        s_k)}}a \\
    &= \bigcup \{
    \pderiv[\emptyset,\emptyset]{\ApplySubst{\sigma_t} {(s_j \cdots s_k)}}a
    \mid \Null (\ApplySubst{\sigma_t}{(s_1 \cdots s_{j-1})})\emptyset \} \\
    &= \bigcup \{
    \pderiv[\emptyset,\emptyset]{\ApplySubst{\sigma_t}{s_j}}a \cdot \ApplySubst{\sigma_t} {(s_{j+1} \cdots s_k)}
    \mid \Null (\ApplySubst{\sigma_t}{(s_1 \cdots s_{j-1})})\emptyset \}
  \end{align*}

  We prove by induction on $s_j$ that its partial derivatives have the form $[r_1, \dots, r_n]$, for
  $n\ge1$, with all $r_i$ of form \textbf{top}.
  All elements of vector $r_n$ are $\sigma_t$ substitution instances of $t$-subexpressions strictly
  before $s_{j+1}$. We write $\vec s$ for $\ApplySubst{\sigma_t} {(s_{j+1} \cdots s_k)}$.

  \textbf{Case }$\Rnull$: does not appear in a partial derivative.

  \textbf{Case }$\Rempty$: the partial derivative has no elements.

  \textbf{Case }$a$: the only possible element of the output is
  $[\Rempty]$ of form \textbf{top}. Obviously, it can be extended to a
  form \textbf{top} expression by appending $ \vec s$.

  \textbf{Case }$r+s$: immediate by induction.

  \textbf{Case }$r\cdot s$: If $\RS = [r_1,\dots,r_n] \in
  \pderiv[\emptyset,\emptyset]{\ApplySubst{\sigma_t}r}a$, then its
  form is according to
  Lemma~\ref{lemma:classification-of-stack-elements-finally} with all elements
  of the vector $r_n$ strictly before $s$, by induction. In that case, $\RS
  \cdot (\ApplySubst{\sigma_t}s) \in
  \pderiv[\emptyset,\emptyset]{\ApplySubst{\sigma_t}{(r\cdot s)}}a$
  where the bottom vector $r_n \cdot  (\ApplySubst{\sigma_t}s)$ has
  form \textbf{top} and is
  composed of elements strictly before $r\cdot s$ and by transitivity
  before $s_{j+1}$.
  Otherwise, if $\Null (\ApplySubst{\sigma_t}r)\emptyset$, then
  $\pderiv[\emptyset,\emptyset]{\ApplySubst{\sigma_t}s}a \subseteq
  \pderiv[\emptyset,\emptyset]{\ApplySubst{\sigma_t}{(r\cdot s)}}a$
  and the claim holds by induction.

  \textbf{Case }$r^*$: If  $\RS = [r_1,\dots,r_n] \in
  \pderiv[\emptyset,\emptyset]{\ApplySubst{\sigma_t}r}a$, then its
  form is according to
  Lemma~\ref{lemma:classification-of-stack-elements-finally} with all elements
  of the vector $r_n$ strictly before $r^*$, by induction.
  In that case, $\RS
  \cdot (\ApplySubst{\sigma_t}{(r^*)}) \in
  \pderiv[\emptyset,\emptyset]{\ApplySubst{\sigma_t}{(r^*)}}a$
  where the bottom vector $r_n \cdot  (\ApplySubst{\sigma_t}r^*)$ has
  form \textbf{top} and is
  composed of elements before $r^*$ and by transitivity
  strictly before $s_{j+1}$.

  \textbf{Case }$\mu x.r$:
  Each stack in
  $\pderiv[\emptyset,\emptyset]{\ApplySubst{\sigma_t}{\mu x.r}}a$ has
  the form  $\RS \PUSH \SINGLETON\Rempty$ and $\RS \in
  \pderiv[\emptyset,\emptyset]{\ApplySubst{(\sigma_t \setminus x)}{r}}a$ .
  The claim holds by induction for the vectors in $\RS$. The last element of the stack has form
  \textbf{top} and it is strictly before $s_{j+1}$ by construction.

  \textbf{Case }$x$: the partial derivative is empty.

  \textbf{Form~rec}.
  As $t_0$ is closed, it must be that for a $t$-sorted vector $s_1\cdots s_k$
  \begin{align*}
    \pderiv[\emptyset,\emptyset]{t_0}a
    &= \pderiv[\emptyset,\emptyset]{\ApplySubst{\sigma_t}{(({\mu x.s_0}) \cdot (s_1\cdots s_k))}}a \\
    &=
    \pderiv[\emptyset,\emptyset]{\ApplySubst{\sigma_t}{({\mu x.s_0})}}a  \cdot \ApplySubst{\sigma_t}(s_1\cdots s_k)
    \cup \pderiv[\emptyset,\emptyset]{\ApplySubst{\sigma_t}{(s_1\cdots s_k)}}a 
  \end{align*}
  because $\ApplySubst{\sigma_t}{\Angle{\mu x.s_0}}$ is nullable.
  A stack in the right argument of the union has form $\mathbf{top}^+$ as in the case
  for form \textbf{top}. It remains to consider 
  \begin{align*}
    \pderiv[\emptyset,\emptyset]{\ApplySubst{\sigma_t}{({\mu x.s_0})}}a
    &= \pderiv[\emptyset,\emptyset]{s_0}{a}
                                      \PUSH
                                      \SINGLETON\Rempty
  \end{align*}
  By Lemma~\ref{lemma:classification-of-stack-elements-finally} each stack 
  $\RS \in \pderiv[{[\mu x.s_0/x]}, {[\Null
    (s_0)\nu[\False/x]/x]}]{\ApplySubst{(\sigma_t \setminus x)}s_0}{a}$ has form
  $\mathbf{top}^+$.

  \textbf{Derivation by }$\varepsilon$. The case analysis for
  $\pderiv[\emptyset,\emptyset]{t_0}\varepsilon$ is analogous to the
  analysis above. The only substantially  different case is the following.

  \textbf{Subcase }$x$: In this case,
  $\pderiv[\sigma,\nu]{x}\varepsilon = \{ \SINGLETON{\ApplySubst\sigma
    x} \}$ where $\sigma (x) = \ApplySubst{\sigma_t}{\mu x.r}$. Hence,
  the derivative has the form \textbf{rec} and thus $\mathbf{rec}.\mathbf{top}^*$.
  \qed
\end{proof}

\end{document}